\documentclass[english]{cccconf}

\usepackage[english]{babel}
\usepackage{theorem}
\theoremheaderfont{\itshape\bfseries}
{\theorembodyfont{\itshape}
\newtheorem{assumption}{\textbf{Assumption}}
\newtheorem{lemma}{\textbf{Lemma}}
\newtheorem{definition}{\textbf{Definition}}
\newtheorem{theorem}{\textbf{Theorem}}

\newtheorem{remark}{\textbf{Remark}}

\newtheorem{problem}{\textbf{Problem}}
}
\usepackage{graphicx}
\usepackage{newtxtext,newtxmath}
\usepackage{amsmath}
\usepackage{amsfonts}
\usepackage{amssymb}
\usepackage{mathrsfs}
\usepackage{xcolor}
\usepackage{graphicx}
\usepackage{tcolorbox}
\usepackage{mdframed}
\usepackage{booktabs}
\usepackage{booktabs}
\usepackage{caption}
\usepackage{supertabular}

\newcommand{\T}{^{\mbox{\tiny T}}}
\newcommand{\R}{\mathbb{R}}

\let\leq\leqslant
\let\geq\geqslant

\newenvironment{proof}[1][Proof]%
{\par\noindent\textit{#1:\ }}%
{\hspace*{\fill} \rule{6pt}{6pt}}
\newenvironment{proof*}[1][Proof]%
{\par\noindent\textit{#1:\ }}{}

\DeclareMathOperator{\diag}{diag}

\DeclareMathOperator{\rank}{rank}

\usepackage{tikz}
\usetikzlibrary{shapes,calc,arrows,patterns,decorations.pathmorphing
	,decorations.markings}
\usetikzlibrary{arrows.meta}

\newenvironment{system}[1]%
{\setlength{\arraycolsep}{0.5mm}\left\{ \; \begin{array}{#1}}%
    {\end{array} \right.}
\newenvironment{system*}[1]%
{\setlength{\arraycolsep}{0.5mm} \begin{array}{#1}}%
  {\end{array}}

\begin{document}
	
	\title{Scale-free Collaborative Protocol Design for Synchronization of Homogeneous and Heterogeneous Discrete-time Multi-agent Systems}
	\author{Donya Nojavanzadeh\aref{wsu}, Zhenwei Liu\aref{neu}, Ali Saberi\aref{wsu}, Anton A. Stoorvogel\aref{ut}}
		\affiliation[wsu]{School of Electrical Engineering and Computer
		Science, Washington State University, Pullman, WA 99164, USA
		\email{donya.nojavanzadeh@wsu.edu; saberi@eecs.wsu.edu}}
	\affiliation[neu]{College of Information Science and
		Engineering, Northeastern University, Shenyang 110819, China
		\email{liuzhenwei@ise.neu.edu.cn}}
	\affiliation[ut]{Department of Electrical Engineering,
		Mathematics and Computer Science, University of Twente, Enschede, The Netherlands
		\email{A.A.Stoorvogel@utwente.nl}}

	\maketitle
	
	\begin{abstract}
		This paper studies synchronization of homogeneous and heterogeneous discrete-time multi-agent systems. A class of linear dynamic protocol design methodology is developed based on localized information exchange with neighbors which does not need any knowledge of the directed network topology and the spectrum of the associated Laplacian matrix. The main contribution of this paper is that the proposed protocols are scale-free and achieve synchronization for arbitrary number of agents.
	\end{abstract}
	
	
	\footnotetext{Z. Liu's work is supported by Nature Science
		Foundation of Liaoning Province under Grant 2019-MS-116.}

\section{Introduction}

The synchronization problem of multi-agent systems (MAS) has attracted
substantial attention during the past decade, due to the wide potential for
applications in several areas such as automotive vehicle control,
satellites/robots formation, sensor networks, and so on. See for
instance the books \cite{ren-book} and \cite{wu-book} or the
survey paper \cite{saber-murray3}.

We identify two classes of multi-agent systems: homogeneous (i.e. agents are identical) and heterogeneous (i.e. agents are non-identical). State synchronization inherently requires homogeneous MAS. On the other hand, for a heterogeneous MAS generically, state synchronization cannot be achieved and focus has been on output synchronization.
For homogeneous MAS state
synchronization based on diffusive full-state coupling has been studied where the
agent dynamics progress from single- and double-integrator dynamics
(e.g.  \cite{saber-murray2}, \cite{ren}, \cite{ren-beard}) to more
general dynamics (e.g. \cite{scardovi-sepulchre}, \cite{tuna1},
\cite{wieland-kim-allgower}). State synchronization based on
diffusive partial-state coupling has also been considered, including static design (\cite{liu-zhang-saberi-stoorvogel-auto} and \cite{liu-zhang-saberi-stoorvogel-ejc}), dynamic design (\cite{kim-shim-back-seo}, \cite{seo-back-kim-shim-iet}, \cite{seo-shim-back}, \cite{su-huang-tac},
\cite{tuna3}), and the design with localized communication (\cite{chowdhury-khalil} and \cite{scardovi-sepulchre}). Recently, scale-free collaborative protocol designs are developed for continuous-time heterogeneous MAS \cite{donya-liu-saberi-stoorvogel-ACC2020} and for homogeneous continues-time MAS subject to actuator saturation \cite{liu-saberi-stoorvogel-donya-cdc2019}.  
For MAS with
discrete-time agents, earlier work can be found in
\cite{saber-murray2,ren-beard,li-zhang,hadjicostis-charalambous,%
	eichler-werner,tuna2} for essentially first and second-order agents,
and in \cite{li-duan-chen,you-xie,hengster-you-lewis-xie,%
	lee-kim-shim,zhou-xu-duan,zhao-park-zhang-shen,wang-saberi-yang,%
	wang-saberi-stoorvogel-grip-yang} for higher-order agents.


In heterogeneous MAS, if the agents have absolute measurements of
their own dynamics in addition to relative information from the network, they are said to be introspective, otherwise, they are called non-introspective. The output synchronization problem for agents with general dynamics has been studied in both introspective and non-introspective cases. For heterogeneous MAS with introspective right-invertible agents, \cite{wang-saberi-yang} 
and 
\cite{yang-saberi-stoorvogel-grip-journal} developed the output and regulated output synchronization results for discrete-time and continuous-time agents. Reference
\cite{li-soh-xie-lewis-tac2019} provided regulated output consensus for both continuous- and discrete-time introspective agents. On the other hand,
for heterogeneous MAS with non-introspective agents, \cite{wieland-sepulchre-allgower} developed an internal model principle based design (see also \cite{grip-saberi-stoorvogel3}) and \cite{grip-yang-saberi-stoorvogel-automatica} considered the output and regulated output synchronization. Reference \cite{chopra-tac} designed a static protocol design for MAS with non-introspective passive agents and \cite{grip-saberi-stoorvogel} provided a purely distributed low-and high-gain based linear time-invariant protocol design for non-introspective homogeneous MAS with linear and nonlinear agents and for non-introspective heterogeneous MAS.

In this paper, we design \textbf{scale-free} collaborative protocols based on localized information exchange among neighbors for synchronization of homogeneous and heterogeneous discrete-time MAS. We study synchronization problem for discrete-time homogeneous MAS with non-introspective agents for both full- and partial-state coupling. Moreover, we deal with output and regulated output synchronization for heterogeneous discrete-time MAS with introspective agents. The protocol design is scale-free, namely:
\begin{itemize}
	\item The design is independent of the information about communication networks such as a lower bound of non-zero eigenvalue of associated Laplacian matrix.
	\item The one-shot protocol design only depends on agent models and does not need any information about communication network and the number of agents.
	\item The synchronization is achieved for any MAS with any number of agents, and any communication network.
\end{itemize}

\subsection*{Notations and definitions}

Given a matrix $A\in \mathbb{R}^{m\times n}$, $A\T$ denotes its
conjugate transpose. A square matrix
$A$ is said to be Schur stable if all its eigenvalues are in the
open unit disc. We denote by
$\diag\{A_1,\ldots, A_N \}$, a block-diagonal matrix with
$A_1,\ldots,A_N$ as its diagonal elements. $A\otimes B$ depicts the
Kronecker product between $A$ and $B$. $I_n$ denotes the
$n$-dimensional identity matrix and $0_n$ denotes $n\times n$ zero
matrix; sometimes we drop the subscript if the dimension is clear from
the context.

To describe the information flow among the agents we associate a \emph{weighted graph} $\mathcal{G}$ to the communication network. The weighted graph $\mathcal{G}$ is defined by a triple
$(\mathcal{V}, \mathcal{E}, \mathcal{A})$ where
$\mathcal{V}=\{1,\ldots, N\}$ is a node set, $\mathcal{E}$ is a set of
pairs of nodes indicating connections among nodes, and
$\mathcal{A}=[a_{ij}]\in \mathbb{R}^{N\times N}$ is the weighted adjacency matrix with non negative elements $a_{ij}$. Each pair in $\mathcal{E}$ is called an \emph{edge}, where
$a_{ij}>0$ denotes an edge $(j,i)\in \mathcal{E}$ from node $j$ to
node $i$ with weight $a_{ij}$. Moreover, $a_{ij}=0$ if there is no
edge from node $j$ to node $i$. We assume there are no self-loops,
i.e.\ we have $a_{ii}=0$. The \emph{weighted in-degree} of a node $i$ is given by $d_{in}(i)=\sum_{j=1}^{N}a_{ij}$. Similarly, the  \emph{weighted out-degree} of a node $i$, is given by $d_{out}(i)=\sum_{j=1}^{N}a_{ji}$. A \emph{path} from node $i_1$ to $i_k$ is a
sequence of nodes $\{i_1,\ldots, i_k\}$ such that
$(i_j, i_{j+1})\in \mathcal{E}$ for $j=1,\ldots, k-1$. A \emph{directed tree} is a subgraph (subset
of nodes and edges) in which every node has exactly one parent node except for one node, called the \emph{root}, which has no parent node. A \emph{directed spanning tree} is a subgraph which is
a directed tree containing all the nodes of the original graph. If a directed spanning tree exists, the root has a directed path to every other node in the tree \cite{royle-godsil}.  

For a weighted graph $\mathcal{G}$, the matrix
$L=[\ell_{ij}]$ with
\[
\ell_{ij}=
\begin{system}{cl}
\sum_{k=1}^{N} a_{ik}, & i=j,\\
-a_{ij}, & i\neq j,
\end{system}
\]
is called the \emph{Laplacian matrix} associated with the graph
$\mathcal{G}$. The Laplacian matrix $L$ has all its eigenvalues in the
closed right half plane and at least one eigenvalue at zero associated
with right eigenvector $\textbf{1}$ \cite{royle-godsil}. Moreover, if the graph contains a directed spanning tree, the Laplacian matrix $L$ has a single eigenvalue at the origin and all other eigenvalues are located in the open right-half complex plane \cite{chebotarev-agaev}.

	\section{Homogeneous MAS with Non-introspective Agents}\label{homo}

Consider a MAS composed of $N$ identical linear time-invariant agents
of the form,
\begin{equation}\label{discr-agent-model}
\begin{system*}{cl}
{x}_i(k+1) &= Ax_i(k) +B u_i(k),  \\
y_i(k) &= Cx_i(k),
\end{system*}\qquad (i=1,\ldots,N)
\end{equation}
where $x_i(k)\in\R^n$, $u_i(k)\in\R^m$, $y_i(k)\in\R^p$ are respectively the
state, input, and output vectors of agent $i$. Meanwhile, \eqref{discr-agent-model} satisfies the following assumption.
\begin{assumption}\label{ass1}
	We assume that
	\begin{itemize}
		\item all eigenvalues of $A$ are in the closed unit disk.
		\item $(A,B,C)$ is stabilizable and detectable.
	\end{itemize}
	
\end{assumption}

The communication network provides each agent with a linear
combination of its own outputs relative to that of other neighboring
agents. In particular, each agent $i\in\{1,\ldots,N\}$ has access to the
quantity,
\begin{equation}\label{hodt-zeta}
\zeta_i(k)=\dfrac{1}{1+d_{in}(i)}\sum_{j=1}^N a_{ij}(y_i(k)-y_j(k)),
\end{equation}
where $a_{ij}\geq 0$, and $a_{ii}=0$ for $i, j \in \{1,\ldots,N\}$. The
topology of the network can be described by a graph $\mathcal{G}$ with
nodes corresponding to the agents in the network and edges given by
the nonzero coefficients $a_{ij}$. In particular, $a_{ij}>0$ implies
that an edge exists from agent $j$ to $i$. The weight of the edge
equals the magnitude of $a_{ij}$. Next we write $\zeta_i$ as
\begin{equation}\label{zeta-y}
\zeta_i(k)=\sum_{j=1}^N d_{ij}(y_i(k)-y_j(k)),
\end{equation}
where $d_{ij}\geq 0$, and we choose
$d_{ii}=1-\sum_{j=1,j\neq i}^Nd_{ij}$ such that $\sum_{j=1}^Nd_{ij}=1$
with $i,j\in\{1,\ldots,N\}$. Note that $d_{ii}$ satisfies $d_{ii}>0$. The weight matrix $D=[d_{ij}]$ is then a so-called, row stochastic matrix. Let $D_{in}=\diag\{d_{in}(i)\}$ with
$d_{in}(i)=\sum_{j=1}^{N}a_{ij}$. Then the relationship between the
row stochastic matrix $D$ and the Laplacian matrix $L$ is
\begin{equation}\label{hodt-LD}
(I+D_{in})^{-1}L=I-D.
\end{equation}

We refer to \eqref{hodt-zeta} as \emph{partial-state coupling} since only part of
the states are communicated over the network. When $C=I$, it means all states are communicated over the network and we call it \emph{full-state coupling}. Then, the original agents are expressed as
\begin{equation}\label{discr-agent-model2}
{x}_i(k+1) = Ax_i(k) +B u_i(k)
\end{equation}
and $\zeta_i(k)$ is rewritten as
\begin{equation}\label{zeta-x}
\zeta_i(k)=\sum_{j=1}^N d_{ij}(x_i(k)-x_j(k)).
\end{equation}

We define the set of graphs $\mathbb{G}^N$ for the network communication topology as following.
\begin{definition}\label{Def1}
Let $\mathbb{G}^N$ denote the set of directed graphs of $N$ agents which contains a directed spanning tree.
\end{definition}

If the graph $\mathcal{G}$ describing the communication topology of
the network contains a directed spanning tree, then it follows from
\cite[Lemma $3.5$]{ren-beard} that the row stochastic matrix $D$ has
a simple eigenvalue at $1$ with corresponding right eigenvector
$\mathbf{1}$ and all other eigenvalues are strictly within the unit
disc. Let $\lambda_1,\ldots,\lambda_N$ denote the eigenvalues of $D$
such that $\lambda_1=1$ and $|\lambda_i|<1$, $i=2,\ldots,N$.

Obviously, state synchronization is achieved if
\begin{equation}\label{synch_org}
\lim_{k\to \infty} (x_i(k)-x_j(k))=0.
\end{equation}
for all $i,j \in {1,...,N}$.

In this paper, we also introduce a localized
information exchange among protocols. In particular, each agent 
$i=1,\ldots, N$ has access to the localized information, denoted by
$\hat{\zeta}_i(k)$, of the form
\begin{equation}\label{eqa1}
\hat{\zeta}_i(k)=\sum_{j=1}^Nd_{ij}(\rho_i(k)-\rho_j(k))
\end{equation}
where $\rho_j(k)\in\mathbb{R}^n$ is a variable produced internally by agent $j$ and to be defined in next sections.


We formulate the following problem for state
synchronization of a homogeneous MAS based on localized information exchange.

\begin{problem}\label{prob1}
	Consider a MAS described by \eqref{discr-agent-model} and \eqref{zeta-y} satisfying Assumption \ref{ass1}. Let
	$\mathbb{G}^N$ be the set of network graphs as defined in Definition \ref{Def1}. Then the \textbf{scalable state synchronization problem based on localized information exchange} is to find, if possible, a linear dynamic controller for each agent $i \in\{1, \dots, N\}$, using only knowledge of the agents model, i.e. $(C,A,B)$, of the form: 
	\begin{equation}\label{protocol1}
	\begin{system}{cl}
	{x}_{i,c}(k+1)&=A_{c}x_{i,c}(k)+B_{c}\zeta_i(k)+C_{c}\hat{\zeta}_i(k),\\
	u_i(k)&=F_cx_{i,c}(k),
	\end{system}
	\end{equation}
	where $\hat{\zeta}_i(k)$ is defined as \eqref{eqa1} with $\rho_i=M_{c}x_{i,c}$, and $x_{i,c}\in\mathbb{R}^{n_i}$, such that state synchronization \eqref{synch_org} is achieved for all initial conditions.
\end{problem}

\noindent{\textbf{Protocol Design}}

Now, we consider state synchronization problem of a homogeneous MAS for both cases of full- and partial-state coupling.

\subsection{Full-state coupling}

In this subsection, we consider state synchronization of MAS with full-state coupling. The design procedure is given in Protocol $1$.

\begin{table}[h]
	\centering
	\captionsetup[table]{labelformat=empty}
	\caption*{Protocol 1: Scale-free collaborative protocol design for homogeneous MAS with full-state coupling}
	\begin{tabular}{p{8cm}}
		\toprule
 We design dynamic collaborative protocols utilizing localized information exchange for agent $i\in\{1,\ldots,N\}$ as
\begin{equation}\label{pscp1}
\begin{system}{cll}
{\eta}_i(k+1) &=& A\eta_i(k)+Bu_i(k)+A{\zeta}_i(k)-A\hat{\zeta}_i(k) \\
u_i(k) &=& - K\eta_i(k),
\end{system}
\end{equation}
 where $K$ is a matrix such that $A-BK$ is Schur stable and $\rho_i$ is a variable produced internally by agent $i$ and is chosen in \eqref{eqa1} as $\rho_i=\eta_i$, therefore each agent has access to the following information:
\begin{equation}\label{info1}
\hat{\zeta}_i(k)=\sum_{j=1}^Nd_{ij}(\eta_i(k)-\eta_j(k)).
\end{equation}
meanwhile, ${\zeta}_i$ is defined in \eqref{zeta-x}. The architecture of the protocol is shown in Figure \ref{homogeneous_full-state}.\\
\bottomrule
\end{tabular}
\end{table} 

\begin{figure}[t]
	\includegraphics[width=8.5cm, height=5cm]{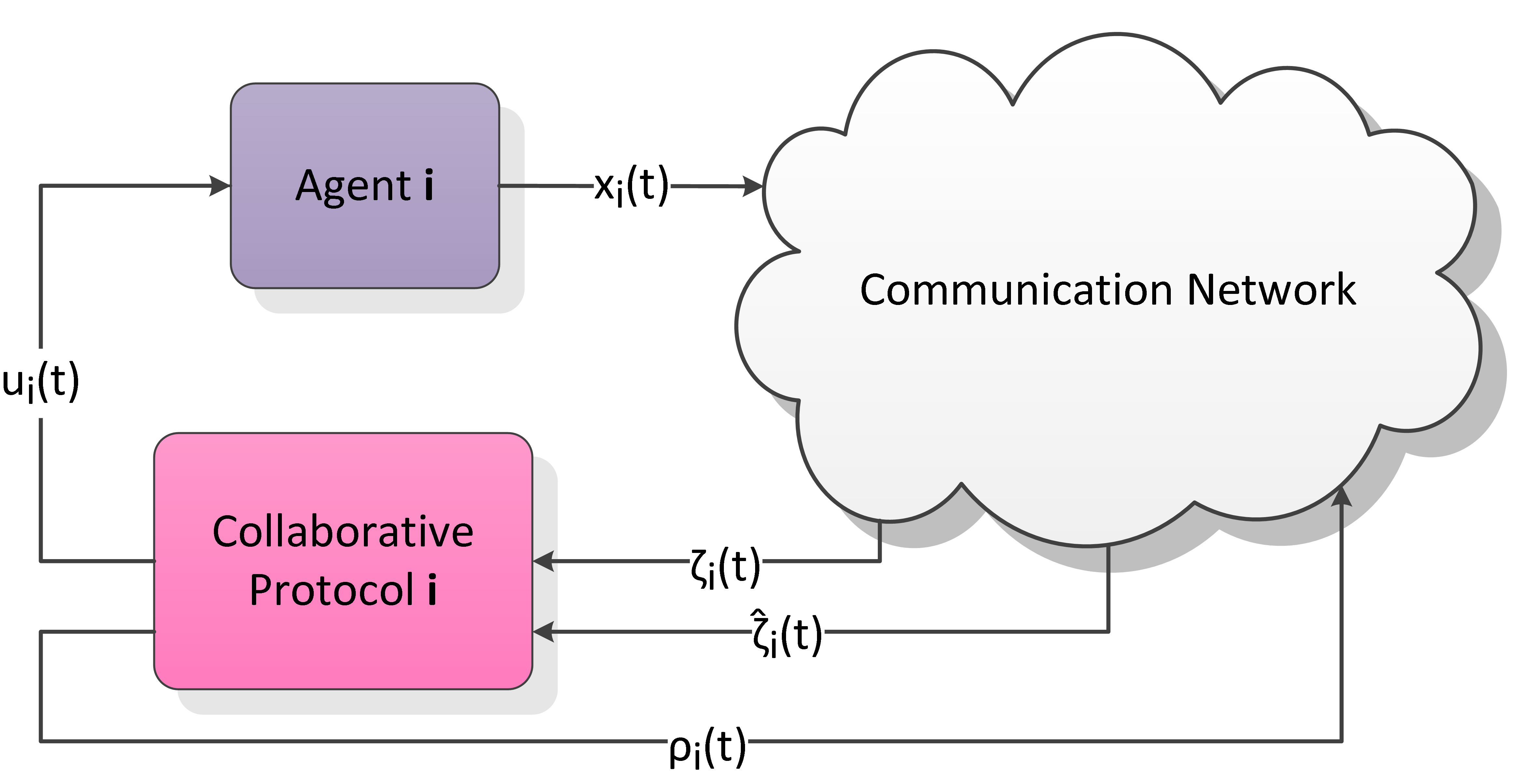}
	\centering
	\caption{Architecture of Protocol $1$}\label{homogeneous_full-state}
\end{figure}
Our formal result is stated in the following theorem.
\begin{theorem}\label{mainthm1}
	Consider a MAS described by \eqref{discr-agent-model2} and \eqref{zeta-x} satisfying Assumption \ref{ass1}. Let
$\mathbb{G}^N$ be the set of network graphs as defined in Definition \ref{Def1}. Then the scalable state synchronization problem based on localized information exchange as stated in Problem
\ref{prob1} is solvable. In particular, the dynamic protocol \eqref{pscp1} solves the state
synchronization problem for any graph
$\mathcal{G} \in \mathbb{G}^N$ with any number of agents $N$.
\end{theorem}

To obtain this result, we recall the following lemma.

\begin{lemma}\label{lem-Dtilde}
	Let a row stochastic matrix $D\in \R^{N\times N}$ be given. We define
	$\tilde{D}\in \R^{(N-1)\times (N-1)}$ as the matrix
	$\tilde{D}=[\tilde{d}_{ij}]$ with
	\[
	\tilde{d}_{ij} = d_{ij}-d_{Nj}.
	\]
	Then the eigenvalues of $\tilde{D}$ are equal to the nonzero
	eigenvalues of $D$.
\end{lemma}

\begin{proof}[Proof of Lemma \ref{lem-Dtilde}]
	We have:
	\[
	I-\tilde{D} = \begin{pmatrix} I & -\textbf{1} \end{pmatrix}
	(I-D) \begin{pmatrix} I \\ 0 \end{pmatrix}
	\]
	Assume that $\lambda$ is a nonzero eigenvalue of $I-D$ with eigenvector
	$x$, then
	\[
	\tilde{x} = \begin{pmatrix} I & -\textbf{1} \end{pmatrix} x
	\]
	where $\textbf{1}$ is a vector with all $1$'s, satisfies,
	\[
	\begin{pmatrix} I & -\textbf{1} \end{pmatrix} (I-D)x =
	\begin{pmatrix} I & -\textbf{1} \end{pmatrix} \lambda x
	=\lambda \tilde{x}
	\]
	and since $(I-D)\textbf{1}=0$ we find that 
	\[
	(I-\tilde{D}) \tilde{x} =
	\begin{pmatrix} I & -\textbf{1} \end{pmatrix} (I-D)x =\lambda \tilde{x}.
	\]
	This shows that  $\lambda$ is an eigenvector of $(I-\tilde{D})$ if
	$\tilde{x}\neq0$. It is easily seen that $\tilde{x}=0$ if and only if
	$\lambda=0$. Conversely if $\tilde{x}$ is an eigenvector of $(I-\tilde{D})$
	with eigenvalue $\lambda$, then it is easily verified that
	\[
	x = (I-D) \begin{pmatrix} I \\ 0 \end{pmatrix} \tilde{x}
	\]
	is an eigenvector of $(I-D)$ with eigenvalue $\lambda$.
\end{proof}\\

\begin{proof}[Proof of Theorem \ref{mainthm1}] Firstly, let $\bar{x}_i(k)=x_i(k)-x_N(k)$ and $\bar{\eta}_i(k)=\eta_i(k)-\eta_N(k)$, we have
	\[
	\begin{system*}{l}
	{\bar{x}}_i(k+1)=A{\bar{x}}_i(k)+B(u_i(k)-u_N(k)),\\
	{\bar{\eta}}_i(k+1)=A\bar{\eta}_i(k)+B(u_i(k)-u_N(k))\\
	\hspace*{1.7cm}+A(\bar{x}_i(k)-\bar{\eta}_i(k))+\sum_{j=1}^{N-1}\tilde{d}_{ij}A(\bar{x}_j(k)-\bar{\eta}_j(k)),\\
	u_i(k)-u_N(k)=-K \bar{\eta}_i(k).
	\end{system*}
	\]
	where $\tilde{d}_{ij}=d_{ij}-d_{Nj}$. Then, we define
	\[
	\bar{x}(k)=\begin{pmatrix}
	\bar{x}_1(k)\\\vdots\\\bar{x}_{N-1}(k)
	\end{pmatrix}  \text{ and } 
	\bar{\eta}(k)=\begin{pmatrix}
	\bar{\eta}_1(k)\\\vdots\\\bar{\eta}_{N-1}(k).
	\end{pmatrix} 	
	\]	
	
	 Based on Lemma \ref{lem-Dtilde}, we have that eigenvalues of $\tilde{D}$ are equal to the
	eigenvalues of $D$ unequal to $1$. Then, we have the following closed-loop system
	\begin{equation}
	\begin{system}{l}
	{\bar{x}}(k+1)=(I\otimes A) \bar{x}(k)-(I\otimes BK)\bar{\eta}(k) \\
	{\bar{\eta}}(k+1)=I\otimes (A-BK) \bar{\eta}(k)\\
	\qquad\qquad\qquad+((I-\tilde{D})\otimes A)(\bar{x}(k)-\bar{\eta}(k))
	\end{system}
	\end{equation}

	Let $e(k)=\bar{x}(k)-\bar{\eta}(k)$, we can obtain  
	\begin{align}
	{\bar{x}}(k+1)&=(I\otimes (A-BK))\bar{x}(k)+(I\otimes BK)e(k)\label{newsystem2}\\
	{e}(k+1)&=(\tilde{D}\otimes A) e(k) \label{newsystem22}
	\end{align}
	
	We have that all eigenvalues of $\tilde{D}$ are in open unit disk. The eigenvalues of $\tilde{D} \otimes A$ are of the form $\lambda_i\mu_j$, with $\lambda_i$ and $\mu_j$ eigenvalues of $\tilde{D}$ and $A$, respectively \cite[Theorem 4.2.12]{horn-johnson}. Since $|\lambda_i|<1$ and $|\mu_j|\leq 1$, we find $\tilde{D}\otimes A$ is asymptotically stable. Then we have
	\begin{equation}
	\lim_{k\to \infty}e_i(k)\to 0
	\end{equation}

	According to the above result, for \eqref{newsystem2} we just need to prove the stability of 
	\begin{equation}\label{systemwol}
	{\bar{x}}(k+1)=(I\otimes (A-BK))\bar{x}(k).
	\end{equation}
	given that $A-BK$ is Schur stable, \eqref{systemwol} is asymptotically stable. Then, we will have
	\[
	\lim_{k\to \infty}\bar{x}_i(k)=\lim_{k\to \infty}(x_i(k)-x_N(k))\to 0
	\] 
	i.e. $x_i(k)\to x_j(k)$ as $k\to \infty$, which proves the result.
\end{proof}

\subsection{Partial-state coupling}
In this subsection, we consider state synchronization of MAS with partial-state coupling. The design procedure is given in Protocol $2$.
 \begin{table}[h]
 	\centering
 	\captionsetup[table]{labelformat=empty}
 	\caption*{Protocol 2: Scale-free collaborative protocol design for homogeneous MAS with partial-state coupling}
 	\begin{tabular}{p{8cm}}
 		\toprule
 We propose the following dynamic protocol with localized information exchange for agent
 $i\in\{1,\ldots,N\}$ as follows:
 \begin{equation}\label{pscp2}
 \begin{system}{cll}
 {\eta}_i(k+1) &=& A\eta_i(k)+Bu_i(k)+A{\hat{x}}_i(k)-A\hat{\zeta}_i(k) \\
  {\hat{x}}_i(k+1) &=& A\hat{x}_i(k)-BK\hat{\zeta}_i(k)+H({\zeta}_i(k)-C\hat{x}_i(k)) \\
 u_i(k) &=& - K\eta_i(k)
 \end{system}
 \end{equation}
 where $K$ is a matrix such that $A-BK$ is Schur stable and $\rho_i$ is chosen as $\rho_i=\eta_i$ in \eqref{eqa1} and with this choice of $\rho_i$, $\hat{\zeta}_i$ is given by:
 \begin{equation}\label{info-par1}
 \hat{\zeta}_i(k)=\sum_{j=1}^Nd_{ij}(\eta_i(k)-\eta_j(k)).
 \end{equation}
 meanwhile, ${\zeta}_i$ is defined in \eqref{zeta-y}. The architecture of the protocol is shown in Figure \ref{homogeneous_partial-state}.\\
 \bottomrule
\end{tabular}
\end{table} 

\begin{figure}[t]
\includegraphics[width=8.5cm, height=5cm]{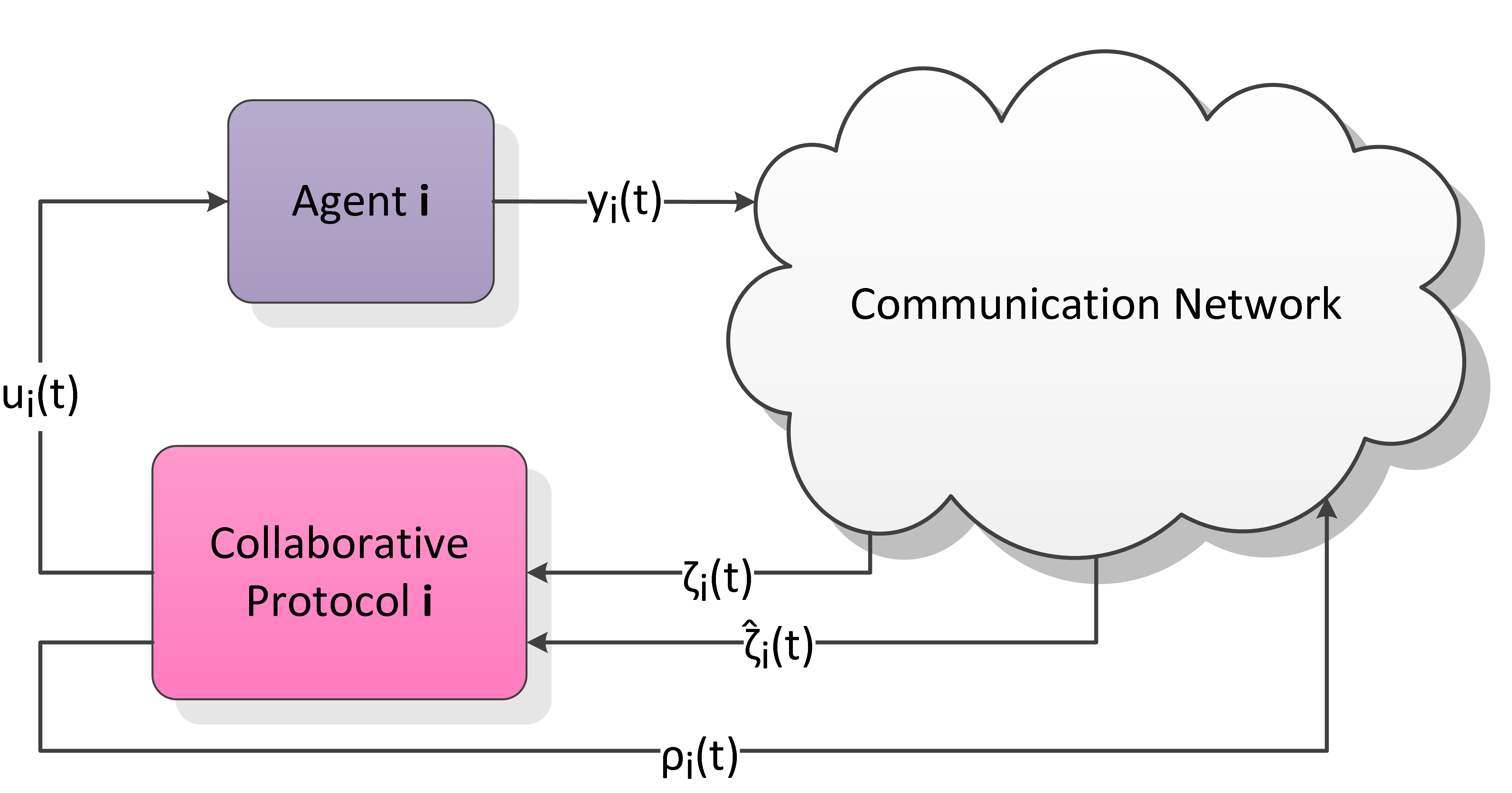}
\centering
\caption{Architecture of Protocol $2$}\label{homogeneous_partial-state}
\end{figure}
 
 Then, we have the following theorem for state synchronization for discrete-time MAS with partial-state coupling.
 
 \begin{theorem}\label{mainthm2}
	Consider a MAS described by \eqref{discr-agent-model} and \eqref{zeta-y} satisfying Assumption \ref{ass1}. Let
$\mathbb{G}^N$ be the set of network graphs as defined in Definition \ref{Def1}. Then the scalable state synchronization problem based on localized information exchange as stated in Problem
 	\ref{prob1} is solvable. In particular, the dynamic protocol \eqref{pscp2} solves the state
 	synchronization problem for any graph
 	$\mathcal{G} \in \mathbb{G}^N$ with any number of agents $N$. 
 \end{theorem}
 
 \begin{proof}[Proof of Theorem \ref{mainthm2}] Let $\bar{x}_i(k)=x_i(k)-x_N(k)$, and $\hat{\bar{x}}_i(k)=\hat{x}_i(k)-\hat{x}_N(k)$, $\bar{\eta}_i(k)=\eta_i(k)-\eta_N(k)$, and $\bar{y}_i(k)=y_i(k)-y_N(k)$, then we  have:
 \begin{equation*}
\begin{system}{cl}
{\bar{x}}_i(k+1) &= A\bar{x}_i(k)-BK\bar{\eta}_i(k)\\
\bar{\eta}_i(k+1) &= -BK\bar{\eta}_i(k)+A{\hat{\bar{x}}}_i(k)-A\sum_{j=1}^{N-1}\tilde{d}_{ij}\bar{\eta}_j(k) \\
{\hat{\bar{x}}}_i(k+1) &= A\hat{\bar{x}}_i(k)-BK\sum_{j=1}^{N-1}\tilde{d}_{ij}(\bar{\eta}_i(k)-\bar{\eta}_j(k))\\
&\hspace{1.3cm}+H(\sum_{j=1}^{N-1}\tilde{d}_{ij}C(\bar{x}_i(k)-\bar{x}_j(k))-C\hat{\bar{x}}_i(k)) \\
\end{system}
\end{equation*}

We also define
\begin{equation*}
\hat{\bar{x}}(k)=\begin{pmatrix}
\hat{\bar{x}}_1(k)\\ \vdots\\ \hat{\bar{x}}_{N-1}(k)
\end{pmatrix}
\end{equation*} 
then, we have the following closed-loop system:
 \begin{equation*}
 \begin{system}{cl}
\bar{x}(k+1) =& (I\otimes A)\bar{x}(k)-(I\otimes BK)\bar{\eta}(k)\\
\bar{\eta}(k+1) =& -(I\otimes BK)\bar{\eta}(k)+(I\otimes A)\hat{\bar{x}}(k)+(\tilde{D}\otimes A)\bar{\eta}(k)\\
\hat{\bar{x}}(k+1) =& (I\otimes(A-HC))\hat{\bar{x}}(k)-((I-\tilde{D})\otimes BK)\bar{\eta}(k)\\
&\hspace{2.35cm}+((I-\tilde{D})\otimes HC)\bar{x}(k)
 \end{system}
 \end{equation*}
By defining $\bar{e}(k)=\bar{x}(k)-\bar{\eta}(k)$ and $\tilde{e}(k)=(\bar{x}(k)-\hat{\bar{x}}(k))-(\tilde{D}\otimes I)\bar{x}(k)$, we can obtain
\begin{equation*}\label{x-e}
\begin{system*}{cl}
\bar{x}(k+1)&=(I\otimes (A-BK))\bar{x}(k)+(I\otimes BK)\bar{e}(k)\\
\bar{e}(k+1)&=(\tilde{D}\otimes A)\bar{e}(k)+(I\otimes A)\tilde{e}(k)\\
\tilde{e}(k+1)&=(I\otimes(A-HC))\tilde{e}(k)
\end{system*}
\end{equation*}

Since $A-BK$, $\tilde{D} \otimes A$, and $A-HC$ are all Schur stable, this system is asymptotically stable. Therefore,
\[
\lim_{k\to \infty}\bar{x}_i(k)=(x_i(k)-x_N(k))\to 0
\] 
i.e. $x_i(k)\to x_j(k)$ as $k\to \infty$, which proves the result.
\end{proof}

	\section{Heterogeneous MAS with Introspective Agents}
	
In this section, we will study a heterogeneous MAS consisting of $N$ non-identical linear agents:
\begin{equation}\label{het-sys}
\begin{system*}{cl}
{x}_i(k+1) &= A_ix_i(k) +B_iu_i(k),  \\
y_i(k) &= C_ix_i(k),
\end{system*}\qquad (i=1,\ldots,N)
\end{equation}
where $x_i\in\mathbb{R}^{n_i}$, $u_i\in\mathbb{R}^{m_i}$ and $y_i\in\mathbb{R}^p$ are the state,
input, output of agent $i$
for $i=1,\ldots, N$.

The agents are introspective, meaning that each agent has access to its own local information. Specifically each agent has access to the quantity
\begin{equation}\label{local}
z_i(k)=C_i^mx_i(k), \quad z_i\in \mathbb{R}^{q_i}
\end{equation}

We also make the following assumption for the agents:

The communication network provides each agent with local information $\zeta_i(k)$ as \eqref{zeta-y}.

\begin{assumption}\label{ass3}
	For agents $i \in \{1,\dots,N\}$, 
	\begin{enumerate}
		\item $(A_i,B_i)$ is stabilizable.
		\item $(C_i, A_i)$ is detectable.
		\item $(C_i,A_i,B_i)$ is right-invertible
		\item $(C_i^m,A_i)$ is detectable. 
	\end{enumerate}
\end{assumption}

\begin{remark}
	 Right-invertibility of a triple $(C_i,A_i,B_i)$ means that given a reference output $y_r(t)$, there exists an initial condition $x_i(0)$ and an input $u_i(t)$ such that $y_i(t)=y_r(t)$ for all non-negative integers $k$. For example, every single-input single-output system is right-invertible, unless its transfer function is identically zero. The definition of right-invertibility can be found in \cite{moylan:sils}.
\end{remark}

The heterogeneous MAS is said to achieve output synchronization if 
\begin{equation}\label{synch_out}
\lim\limits_{k \to\infty}(y_i(k)-y_j(k))=0, \quad\text{for $i,j \in \{1, \dots ,N\}$}.
\end{equation}

First, we formulate scalable output synchronization problem for heterogeneous networks as follows:

\begin{problem}\label{prob_out_sync}
	Consider a heterogeneous MAS described by agent models \eqref{het-sys} and local information \eqref{local}, satisfying Assumption \ref{ass3} and associated network communication \eqref{zeta-y}. Let
	$\mathbb{G}^N$ be the set of network graphs as defined in Definition \ref{Def1}. The \textbf{scalable output synchronization problem based on localized information exchange} is to find, if possible, a linear dynamic controller for each agent $i \in\{1, \dots, N\}$, using only knowledge of the agent model, i.e. $(C_i,A_i,B_i)$, of the form: 
	\begin{equation}\label{out_dyn}
	\begin{system}{cl}
	{x}_{i,c}(k+1)&=A_{i,c}x_{i,c}(k)+B_{i,c}\zeta_i(k)+C_{i,c}\hat{\zeta}_i(k)+D_{i,c}z_i(k),\\
	u_i(k)&=E_{i,c}x_{i,c}(k)+F_{i,c}\zeta_i(k)+G_{i,c}\hat{\zeta}_i(k)+H_{i,c}z_i(k),
	\end{system}
	\end{equation}
	where $\hat{\zeta}_i(k)$ is defined as \eqref{eqa1} with $\rho_i=N_{i,c}x_{i,c}(k)$, and $x_{i,c}(k)\in\mathbb{R}^{n_i}$, 
	such that for all initial conditions the output synchronization \eqref{synch_out} is achieved for any graph $\mathcal{G}\in \mathbb{G}^N$ with any number of agents $N$.
\end{problem}

	Since $(C_i^m,A_i)$ is detectable for $i \in\{1, \dots, N\}$, one can simply asymptotically stabilize individual agents by utilizing $z_i$, without any communication among agents, and hence achieve output synchronization with zero synchronization trajectory, that is $\lim\limits_{k \to \infty} y_i(k)=0, i \in\{1, \dots, N\}$. However, such a case is not of interest in this paper and our aim is to achieve output synchronization with nontrivial synchronization trajectory.\\
	
		 Next, we consider regulated output synchronization where the agent outputs converge to a priori given trajectory generated by a so-called exosystem.

	
	The synchronized trajectory $y_r(k)$ is given by an exosystem as:
	\begin{equation}\label{exo}
	\begin{system*}{cl}
	{x}_r(k+1)&=A_rx_r(k), \quad x_r(0)=x_{r0},\\
	y_r(k)&=C_rx_r(k),
	\end{system*}
	\end{equation}
	where $x_r \in\mathbb{R}^r$ and $y_r\in\mathbb{R}^p$. We make the following assumption about the exosystem \eqref{exo}.
	
	\begin{assumption}\label{ass-exo}
		For the exosystem \eqref{exo},
		\begin{enumerate}
			\item $(C_r, A_r)$ is observable.
			\item  All the eigenvalues of $A_r$ are in the closed unit disc.
		\end{enumerate}
	\end{assumption}
	
	The heterogeneous MAS is said to achieve regulated output synchronization if 
	\begin{equation}\label{reg_synch_out}
	\lim\limits_{k\to\infty}(y_i(k)-y_r(k))=0, \quad\text{for $i \in \{1, \dots ,N\}$}.
	\end{equation}
	
	We assume a nonempty subset $\mathscr{C}$ of the agents which have access to their output relative to the output of the exosystem. In other word, each agent $i$ has access to the quantity 
	\begin{equation}
	\Psi_i(k)=\iota_i(y_i(k)-y_r(k)), \qquad \iota_i=\begin{system}{cl}
	1, \quad i\in \mathscr{C},\\
	0, \quad i\notin \mathscr{C}.
	\end{system}
	\end{equation}
	By combining this with \eqref{zeta-y}, the information exchange among agents is given by
	
	\begin{equation}\label{zetabar}
	\bar{\zeta}_i(k)=\sum_{j=1}^{N}a_{ij}(y_i(k)-y_j(k))+\iota_i(y_i(k)-y_r(k)).
	\end{equation}
	To guarantee that each agent get the information from the exosystem, we need to make sure that there exist a path from node set $\mathscr{C}$ to each node.  Therefore, we define the following set of graphs.
	
	\begin{definition}\label{def_rootset}
		Given a node set $\mathscr{C}$, we denote by $\mathbb{G}_{\mathscr{C}}^N$ the set of all graphs with $N$ nodes containing the node set $\mathscr{C}$, such that every node of the network graph $\mathcal{G}\in\mathbb{G}_\mathscr{C}^N$ is a member of a directed tree
		which has its root contained in the node set $\mathscr{C}$.
	\end{definition}
	
	\begin{remark}
		Note that Definition \ref{def_rootset} does not require necessarily the existence of directed spanning tree.
	\end{remark}

	In the following, we will refer to the node set $\mathscr{C}$ as root set in view of Definition \ref{def_rootset}.
	For any graph $\mathbb{G}_\mathscr{C}^N$, with the Laplacian matrix $L$, we define the expanded Laplacian matrix as: 
	\[
	\bar{L}=L+diag\{\iota_i\}=[\bar{\ell}_{ij}]_{N \times N}
	\]
		which is not a regular Laplacian matrix associated to the graph, since the sum of its rows need not be zero. We know that Definition \ref{def_rootset}, guarantees that all the eigenvalues of $\bar{L}$, have positive real parts. In particular matrix $\bar{L}$ is invertible.
	In terms of the coefficients of the expanded Laplacian matrix $\bar{L}$, $\bar{\zeta}_i$ in \eqref{zetabar} can be rewritten as:
		\begin{multline}\label{zetabar2}
	\bar{\zeta}_i(k)=\frac{1}{2+d_{in}(i)}\sum_{j=1}^{N}\bar{\ell}_{ij}(y_j(k)-y_r(k))\\=y_i(k)-y_r(k)-\sum_{j=1}^{N}\bar{d}_{ij}(y_j(k)-y_r(k))
	\end{multline}
	and we define
	\begin{equation}\label{bar-D}
	\bar{D}=I-(2I+D_{in})^{-1}\bar{L}.
	\end{equation}
	
	It is easily verified that the matrix $\bar{D}$ is a matrix with all elements non negative and the sum of each row is less than or equal to $1$. The matrix $\bar{D}$ has all eigenvalues in the open unit disk if and only if 	every node of the network graph $\mathcal{G}$ is a member of a
	directed tree which has its root contained in the set $\mathscr{C}$ \cite[Lemma 1]{liu2018regulated}.
	
    We also define $\check{\zeta}_i(k)$ as:
	\begin{equation}\label{info2}
	\check{\zeta}_i(k)=\frac{1}{2+d_{in}(i)}\sum_{j=1}^N\bar{\ell}_{ij}\rho_j=\rho_i-\sum_{j=1}^{N}\bar{d}_{ij}\rho_j
	\end{equation}

	Now we formulate the problem of scalable regulated output synchronization.
	
	\begin{problem}\label{prob_reg_sync}
		Consider a heterogeneous MAS described by agent models \eqref{het-sys}, local information \eqref{local} satisfying Assumption \ref{ass3} and the associated exosystem \eqref{exo} satisfying Assumption \ref{ass-exo}. Let a set of nodes $\mathscr{C}$ be given which defines the set $\mathbb{G}^N_\mathscr{C}$. Let the associated network communication be given by \eqref{zetabar2}. The \textbf{scalable regulated output synchronization problem based on localized information exchange} is to find, if possible, a linear dynamic controller for each agent $i \in\{1, \dots, N\}$, using only knowledge of the agent model, i.e. $(C_i, A_i, B_i)$, of the form:
		\[
		\begin{system}{cl}
		{x}_{i,c}(k+1)&=A_{i,c}x_{i,c}(k)+B_{i,c}\bar{\zeta}_i(k)+C_{i,c}\check{\zeta}_i(k)+D_{i,c}z_i(k),\\
		u_i(k)&=E_{i,c}x_{i,c}(k)+F_{i,c}\bar{\zeta}_i(k)+G_{i,c}\check{\zeta}_i(k)+H_{i,c}z_i(k),
		\end{system}
		\]
		where $\check{\zeta}_i(k)$ is defined as \eqref{info2} with $\rho_i=M_{i,c}x_{i,c}(k)$, and $x_{c,i}(k)\in\mathbb{R}^{n_i}$, 
		such that for all initial conditions and for any $x_{r0}$, the regulated output synchronization \eqref{reg_synch_out} is achieved for any $N$ and any graph $\mathscr{G}\in\mathbb{G}^N_\mathscr{C}$.
	\end{problem}

	To obtain our results, first we design a pre-compensator to make all the agents almost identical. This process is proved by detail in \cite{wang-saberi-yang}. Next, we show that output synchronization with respect to the new almost identical models can be achieved using the controller introduced in section \ref{homo} which is based on localized information exchange. It is worth to note that the designed collaborative protocols are scale-free since they do not need any information about the communication graph, other agent models, or number of agents.

\subsection{Output synchronization}\label{OS}

For solving output synchronization problem for heterogeneous network of $N$ agents \eqref{het-sys}, first we recall a critical lemma as stated in \cite{wang-saberi-yang}.

\begin{lemma}\label{lem-homo}
	Consider the heterogeneous network of $N$ agents \eqref{het-sys} with local information \eqref{local}. Let Assumption \ref{ass3} hold and let $\bar{n}_d$ denote the maximum order of infinite zeros of $(C_i,A_i, B_i), i \in \{1, \dots, N\}$. Suppose a triple $(C, A, B)$ is given such that 
	\begin{enumerate}
		\item $\rank(C)=p$
		\item $(C, A, B)$ is invertible of uniform rank $n_q\ge\bar{n}_d$, and has no invariant zeros.
	\end{enumerate}
	Then for each agent $i \in \{1, \dots, N\}$, there exists a pre-compensator of the form 
	\begin{equation}\label{pre_con}
	\begin{system}{cl}
	\xi_i(k+1)&=A_{i,h}\xi_i(k)+B_{i,h}z_i(k)+E_{i,h}v_i(k),\\
	u_i(k)&=C_{i,h}(k)\xi_i(k)+D_{i,h}v_i(k),
	\end{system}
	\end{equation} 
	such that the interconnection of \eqref{het-sys} and \eqref{pre_con} can be written in the following form:
	\begin{equation}\label{sys_homo}
	\begin{system*}{cl}
	{\bar{x}}_i(k+1)&=A\bar{x}_i(k)+B(v_i(k)+d_i(k)),\\
	{y}_i(k)&=C\bar{x}_i(k),
	\end{system*}
	\end{equation} 
	where
	$d_i$ is generated by 
	\begin{equation}\label{sys-d}
	\begin{system*}{cl}
	\omega_i(k+1)&=A_{i,s}\omega_i(k),\\
	d_i(k)&=C_{i,s}\omega_i(k),
	\end{system*}
	\end{equation}
	for $i\in\{1,\hdots,N\}$, where $A_{i,s}$ is Schur stable.
\end{lemma}
\begin{proof}
	The proof of Lemma \ref{lem-homo} is given in \cite[Appendix A.1]{wang-saberi-yang} by explicit construction of a pre-compensator of the form \eqref{pre_con}.
\end{proof}
\begin{remark}
	We would like to make several observations:
	\begin{enumerate}
		\item The property that the triple $(C, A, B)$ is invertible and has no invariant zero implies that $(A, B)$ is controllable and $(C, A)$ is observable.
		\item The triple $(C, A, B)$ is arbitrarily assignable as long as the conditions are satisfied. In particular, one can choose the eigenvalues of $A$ in arbitrary desired place. 
	\end{enumerate}
\end{remark}

Lemma \ref{lem-homo} shows that we can design a pre-compensator based on local information $z_i$ to transform the nonidentical agents to almost identical models given by \eqref{sys_homo} and \eqref{sys-d}. The compensated model has the same model for each agent except for different exponentially decaying signals $d_i$ in the range space of $B$, generated by \eqref{sys-d}.\\

\noindent{\textbf{Protocol Design}}

Now, we design collaborative protocols to solve the scalable output synchronization problem as stated in Problem \ref{prob_out_sync} in two steps. The design procedure is given in Protocol $3$.\\

 \begin{table}[h]
	\centering
	\captionsetup[table]{labelformat=empty}
	\caption*{Protocol 3: Scale-free collaborative protocol design for output synchronization of heterogeneous MAS }
	\begin{tabular}{p{8cm}}
		\toprule
\begin{itemize}
\item\emph{step 1:} First, we choose design parameters $(C,A,B)$ such that
\begin{enumerate}
	\item $\rank(C)=p$
	\item $(C, A, B)$ is invertible of uniform rank $n_q\ge\bar{n}_d$, and has no invariant zeros.
	\item eigenvalues of $A$ are in closed unit disc.
\end{enumerate}
Next, given $(C,A,B)$ and agent models \eqref{het-sys} we design pre-compensators as \eqref{pre_con} for $i\in\{1,...,N\}$ using Lemma \ref{lem-homo}. Then, by applying \eqref{pre_con} to agent models we get the compensated agents as \eqref{sys_homo} and \eqref{sys-d}.
\item\emph{step 2:} In this step, we design dynamic collaborative protocols based on information exchange for compensated agents \eqref{sys_homo} and \eqref{sys-d} as follows.
\[
\begin{system}{cll}
{\eta}_i(k+1) &=& A\eta_i(k)+Bv_i(k)+A{\hat{x}}_i(k)-A\hat{\zeta}_i(k) \\
{\hat{x}}_i(k+1) &=& A\hat{x}_i(k)-BK\hat{\zeta}_i(k)+H({\zeta}_i(k)-C\hat{x}_i(k)) \\
v_i(k) &=& - K\eta_i(k),
\end{system}
\]
where $K$ is a matrix such that $A-BK$ is Schur stable and and $\rho_i$ is chosen as $\rho_i=\eta_i$ in \eqref{eqa1} and with this choice of $\rho_i$, $\hat{\zeta}_i$ is given by:
\begin{equation}\label{info-par2}
\hat{\zeta}_i(k)=\sum_{j=1}^Nd_{ij}(\eta_i(k)-\eta_j(k)).
\end{equation}
meanwhile, ${\zeta}_i$ is defined in \eqref{zeta-y}.

\item\emph{step 3:} Finally, we combine the designed collaborative protocol for homogenized network in step $2$ with pre-compensators in step $1$ to get our protocol as:
\begin{equation}\label{pscp2final}
\begin{system}{cl}
{\xi}_i(k+1)&=A_{i,h}\xi_i(k)+B_{i,h}z_i(k)-E_{i,h}K\eta_i(k),\\
\hat{x}_i(k+1)&=A\hat{x}_i(k)-BK\hat{\zeta}_i(k)+H(\zeta_i(k)-C\hat{x}_i(k))\\
\eta_i(k+1)&=(A-BK)\eta_i(k)+A\hat{x}_i(k)-A\hat{\zeta}_i(k)\\
u_i(k)&=C_{i,h}\xi_i(k)-D_{i,h}K\eta_i(k),
\end{system}
\end{equation}
where $H$ and $K$ are matrices chosen in step $2$. 
\end{itemize}
 The architecture of the protocol is shown in Figure \ref{heterogeneous_output}.\\
\bottomrule
\end{tabular}
\end{table}

Then, we have the following theorem for output synchronization of heterogeneous MAS.
 \begin{theorem}\label{thm_out_syn}
	Consider a heterogeneous MAS described by agent models \eqref{het-sys} and local information \eqref{local} satisfying Assumption \ref{ass3} and associated network communication \eqref{zeta-y} and \eqref{info-par2}. Let
	$\mathbb{G}^N$ be the set of network graphs as defined in Definition \ref{Def1}. Then the scalable output synchronization problem based on localized information exchange as stated in Problem \ref{prob_out_sync} is solvable. In particular, the dynamic protocol \eqref{pscp2final} solves the scalable output synchronization problem for any graph $\mathcal{G}\in \mathbb{G}^N$ with any number of agents $N$. 
\end{theorem} 
\begin{proof}[Proof of Theorem \ref{thm_out_syn}] Let $
	\bar{x}_i^o(k)=\bar{x}_i(k)-\bar{x}_N(k),\
	y_i^o(k)=y_i(k)-y_N(k),\
	\hat{x}_i^o(k)=\hat{x}_i(k)-\hat{x}_N(k), \text{ and }
	\eta_i^o(k)=\eta_i(k)-\eta_N(k),
	$
	we have 
	\[
	\begin{system*}{ll}
	\bar{x}_i^o(k+1)&=A\bar{x}_i^o(k)+B(v_i(k)-v_N(k)+d_i(k)-d_N(k)),\\
	{y}_i^o(k)&=C\bar{x}_i^o(k),\\
	\bar{\zeta}_i(k)&=\zeta_i(k)-\zeta_N(k)={y}_i^o(k)-\sum_{j=1}^{N-1}\tilde{d}_{ij}{y}_j^o(k),\\
	\hat{x}_i^o(k+1)&=A\hat{x}_i^o(k)-BK(\hat{\zeta}_i(k)-\hat{\zeta}_N(k))+H(\bar{\zeta}_i(k)-C\hat{x}_i^o(k))\\
	\eta_i^o(k+1)&=(A-BK)\eta_i^o+A\hat{x}_i^o-A\sum_{j=1}^{N-1}\tilde{d}_{ij}{\eta}_j^o
	\end{system*}
	\]
	where $\tilde{d}_{ij}=d_{ij}-d_{Nj}$ for $i,j=1,\cdots,N-1$.
	We define
	\begin{equation*}
	\bar{x}=\begin{pmatrix}
	\bar{x}_1^o\\ \vdots\\ \bar{x}_{N-1}^o
	\end{pmatrix},\hat{x}=\begin{pmatrix}
	\hat{x}_1^o\\ \vdots\\ \hat{x}_{N-1}^o
	\end{pmatrix},\eta=\begin{pmatrix}
	\eta_1^o\\ \vdots\\ \eta_{N-1}^o
	\end{pmatrix},d=\begin{pmatrix}
	d_1\\ \vdots\\ d_N\end{pmatrix},\omega=\begin{pmatrix}
	\omega_1\\ \vdots\\ \omega_N\end{pmatrix}
	\end{equation*} 
	then, given that $\Pi=\begin{pmatrix}
	I&-\mathbf{1}
	\end{pmatrix}$, we have the following closed-loop system:
	\[
	\begin{system}{cl}
	\bar{x}(k+1)&=(I\otimes A)\bar{x}(k)-(I\otimes BK )\eta+(\Pi\otimes B)d(k)\\
	\hat{x}(k+1)&=(I\otimes (A-HC))\hat{x}(k)-((I-\tilde{D})\otimes BK )\eta(k)\\
&\hspace*{2.8cm}+((I-\tilde{D})\otimes HC)\bar{x}(k)\\
	\eta(k+1)&=(I\otimes(A-BK )-(I-\tilde{D})\otimes A)\eta(k)+(I\otimes A)\hat{x}(k)
	\end{system}
	\]

\begin{figure}[t]
	\includegraphics[width=8.5cm, height=5cm]{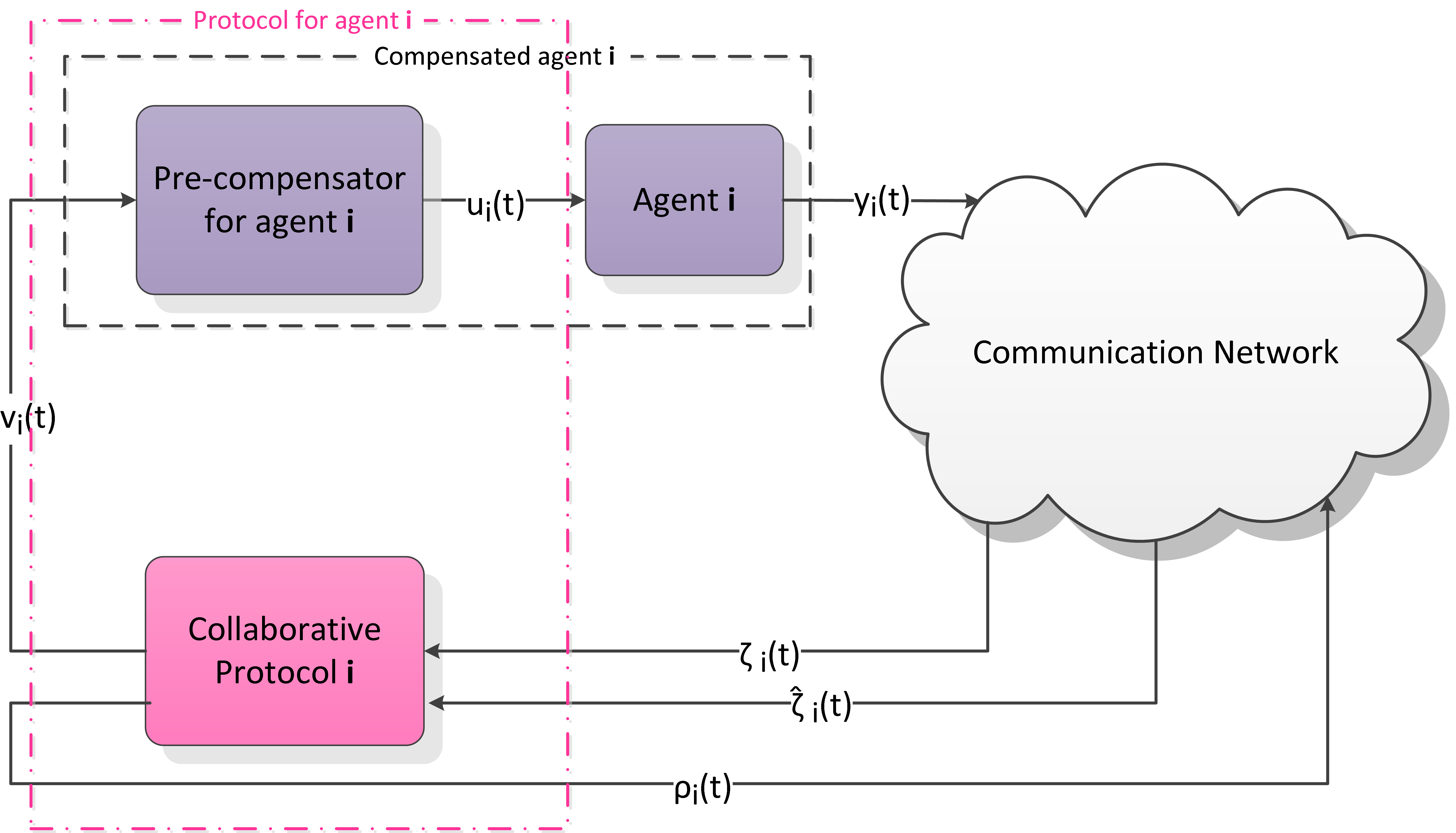}
	\centering
	\caption{Architecture of Protocol $3$}\label{heterogeneous_output}
\end{figure}
	By defining $e(k)=\bar{x}(k)-\eta(k)$ and $\bar{e}(k)=\bar{x}(k)-\hat{x}(k)-(\tilde{D}\otimes I)\bar{x}(k)$, we can obtain
	\[
\begin{system*}{cl}
	\bar{x}(k+1)&=(I\otimes (A-BK))\bar{x}(k)+(I\otimes BK)e(k)\\
	&\hspace*{4.15cm}+(\Pi\otimes B)C_s\omega(k)\\
	e(k+1)&=(\tilde{D}\otimes A)e(k)+(I\otimes A)\bar{e}(k)+(\Pi\otimes B)C_s\omega(k)\\
	\bar{e}(k+1)&=(I\otimes(A-HC))\bar{e}(k)+((I-\tilde{D})\Pi\otimes B)C_s\omega(k)\\
	\omega(k+1)&=A_s\omega(k)
\end{system*}
\]
%
where $A_s=diag\{A_{i,s}\}$ and $C_s=diag\{C_{i,s}\}$. Since $A-BK$, $\tilde{D}\otimes A$, $A-HC$, and $A_s$ are Schur stable, the system is asymptotically stable. Therefore,
\[
\lim_{k\to \infty}\bar{x}_i(k)=(x_i(k)-x_N(k))\to 0
\] 
i.e. $x_i(k)\to x_j(k)$ as $k\to \infty$, which proves the result.
\end{proof}

\subsection{Regulated output synchronization}

Like output synchronization, for solving the regulated output synchronization problem for a network of $N$ agents \eqref{het-sys}, our design procedure consists of three steps.

To obtain our results for regulated output synchronization, we need the following lemma from \cite{wang-saberi-yang}.

\begin{lemma}[\cite{wang-saberi-yang}]\label{lem-exo}
	There exists another exosystem given by:
	\begin{equation}\label{exo-2}
	\begin{system*}{cl}
	\check{x}_r(k+1)&=\check{A}_r\check{x}_r(k), \quad \check{x}_r(0)=\check{x}_{r0}\\
	y_r(k)&=\check{C}_r\check{x}_r(k),
	\end{system*}
	\end{equation}
	such that for all $x_{r0} \in \mathbb{R}^r$, there exists $\check{x}_{r0}\in \mathbb{R}^{\check{r}}$ for which \eqref{exo-2} generate exactly the same output $y_r$ as the original exosystem \eqref{exo}. Furthermore, we can find a matrix $\check{B}_r$ such that the triple $(\check{C}_r,\check{A}_r,\check{B}_r)$ is invertible, of uniform rank $n_q$, and has no invariant zero, where $n_q$ is an integer greater than or equal to maximal order of infinite zeros of $(C_i,A_i,B_i), i\in \{1,...,N\}$ and all the observability indices of $(C_r, A_r)$. Note that the eigenvalues of $\check{A}_r$ consists of all eigenvalues of $A_r$ and additional zero eigenvalues. 
\end{lemma}

\subsection*{\textbf{Protocol Design}}
Now, we design collaborative protocols to solve the scalable regulated output synchronization problem as stated in Problem \ref{prob_reg_sync} in two steps. The design procedure is given in Protocol $4$.\\

 \begin{table}[h]
	\centering
	\captionsetup[table]{labelformat=empty}
	\caption*{Protocol 4: Scale-free collaborative protocol design for regulated output synchronization of heterogeneous MAS}
	\begin{tabular}{p{8cm}}
		\toprule
First in \textit{step $1$}, after choosing appropriate $(\check{C}_r,\check{A}_r,\check{B}_r)$, we design pre-compensators like \textit{step 1} of previous section. Next, in \textit{step 2} we design dynamic collaborative protocols based on localized information exchange for almost identical agents \eqref{sys_homo} and \eqref{sys-d} for $i\in\{1,...,N\}$ as follows:
\[
\begin{system}{cl}
{\eta}_i(k+1) =& \check{A}_r\eta_i(k)+\check{B}_rv_i(k)+\check{A}_r{\hat{x}}_i(k)-\check{A}_r\check{\zeta}_i(k)\\
{\hat{x}}_i(k+1) =& \check{A}_r\hat{x}_i(k)+H({\bar{\zeta}}_i(k)-\check{C}_r\hat{x}_i(k))-\check{B}_rK\check{\zeta}_i(k)\\
v_i(k)=&-K\eta_i(k),
\end{system}
\]

where $H$ and $K$ are matrices such that $\check{A}_r-H\check{C}_r$ and $\check{A}_r-\check{B}_rK$ are Schur stable.
The network information $\check{\zeta}_i$ is defined as \eqref{info2} and $\bar{\zeta}_i$ is defined as \eqref{zetabar2}. Like design procedure in the previous subsection, we combine the designed dynamic collaborative protocols and pre-compensators to get the final protocol as:
\begin{equation}\label{pscp3final}
\begin{system}{cl}
{\xi}_i(k+1)&=A_{i,h}\xi_i(k)+B_{i,h}z_i(k)-E_{i,h}K\eta_i(k),\\
{\hat{x}}_i(k+1) =& \check{A}_r\hat{x}_i(k)+H({\bar{\zeta}}_i(k)-\check{C}_r\hat{x}_i(k))-\check{B}_rK\check{\zeta}_i(k)\\
{\eta}_i(k+1) =& (\check{A}_r-\check{B}_rK)\eta_i(k)+\check{A}_r{\hat{x}}_i(k)-\check{A}_r\check{\zeta}_i(k)\\

u_i(k)&=C_{i,h}\xi_i(k)-D_{i,h}K\eta_i(k),
\end{system}
\end{equation}
where $H$ and $K$ are matrices as defined in step $2$. \\

The architecture of the protocol is shown in Figure \ref{heterogeneous_regulated}.\\
\bottomrule
\end{tabular}
\end{table}

Then, we have the following theorem for regulated output synchronization of heterogeneous MAS.

\begin{theorem}\label{thm_reg_out_syn}
Consider a heterogeneous MAS described by agent models \eqref{het-sys}, local information \eqref{local} satisfying Assumption \ref{ass3} and the associated exosystem \eqref{exo} satisfying Assumption \ref{ass-exo}. Let a set of nodes $\mathscr{C}$ be given which defines the set $\mathbb{G}^N_\mathscr{C}$. Let the associated network communication be given by \eqref{zetabar2} and \eqref{info2}. Then, the scalable regulated output synchronization problem based on localized information exchange as defined in Problem \ref{prob_reg_sync} is solvable. In particular, the dynamic protocol \eqref{pscp3final} solves the scalable regulated output synchronization problem for any graph
	$\mathscr{G}\in\mathbb{G}^N_\mathscr{C}$ with any number of agents $N$. 
\end{theorem}

\begin{proof}[Proof of Theorem \ref{thm_reg_out_syn}] Following Lemma \ref{lem-homo}, we can design a pre-compensator \eqref{pre_con}, for each agent, $i \in \{1,..., N\}$, such that the interconnection of \eqref{het-sys} and \eqref{pre_con}  is as:
	
	\begin{equation}\label{sys_reg_homo}
	\begin{system*}{cl}
	\bar{x}_i(k+1)&=\check{A}_r\bar{x}_i(k)+\check{B}_r(v_i(k)+d_i(k)),\\
	{y}_i(k)&=\check{C}_r\bar{x}_i(k),
	\end{system*}
	\end{equation}
	where $d_i$ is given by \eqref{sys-d}.
	
	First, let $\tilde{x}_i=\bar{x}_i-\check{x}_r$.
	We define
	\begin{equation*}
	\tilde{x}=\begin{pmatrix}
	\tilde{x}_1\\ \vdots\\ \tilde{x}_N
	\end{pmatrix},\hat{x}=\begin{pmatrix}
	\hat{x}_1\\ \vdots\\ \hat{x}_N
	\end{pmatrix},\eta=\begin{pmatrix}
	\eta_1\\ \vdots\\ \eta_N
	\end{pmatrix},d=\begin{pmatrix}
	d_1\\ \vdots\\ d_N\end{pmatrix},\omega=\begin{pmatrix}
	\omega_1\\ \vdots\\ \omega_N\end{pmatrix}
	\end{equation*}
	then we have the following closed-loop system
\begin{equation*}
\begin{system}{cl}
\tilde{x}(k+1) =& (I\otimes \check{A}_r)\tilde{x}(k)-(I\otimes \check{B}_rK)\eta(k)+(I\otimes\check{B}_r)d(k)\\
{\eta}(k+1) =& -(I\otimes \check{B}_rK){\eta}(k)+(I\otimes \check{A}_r)\hat{{x}}(k)+(\bar{D}\otimes \check{A}_r){\eta}(k)\\
\hat{{x}}(k+1) =& [I\otimes(\check{A}_r-H\check{C}_r)]\hat{{x}}(k)-[(I-\bar{D})\otimes \check{B}_rK]{\eta}(k)\\
&\hspace{3.2cm}+[(I-\bar{D})\otimes H\check{C}_r]\tilde{x}(k)
\end{system}
\end{equation*}
By defining ${e}(k)=\tilde{x}(k)-{\eta}(k)$ and $\tilde{e}(k)=(\tilde{x}(k)-\hat{{x}}(k))-(\bar{D}\otimes I)\tilde{x}(k)$, we can obtain
\begin{equation*}\label{x-e3}
\begin{system*}{cl}
\tilde{x}(k+1)&=[I\otimes (\check{A}_r-\check{B}_rK)]\tilde{x}(k)+(I\otimes \check{B}_rK){e}(k)\\
&\hspace*{4.5cm} +(I\otimes\check{B}_r)C_s\omega(k)\\
{e}(k+1)&=(\bar{D}\otimes \check{A}_r){e}(k)+(I\otimes \check{A}_r)\tilde{e}(k)+(I\otimes\check{B}_r)C_s\omega(k)\\
\tilde{e}(k+1)&=[I\otimes(\check{A}_r-H\check{C}_r)]\tilde{e}(k)+((I-\bar{D})\otimes\check{B}_r)C_s\omega(k)\\
\omega(k+1)&=A_s\omega(k)
\end{system*}
\end{equation*}
where $A_s=diag\{A_{i,s}\}$ and $C_s=diag\{C_{i,s}\}$. Since $\check{A}_r-\check{B}_rK$, $\bar{D}\otimes\check{A}_r$, $\check{A}_r-H\check{C}_r$, and $A_s$ are Schur stable, the system is asymptotically stable. Therefore,
\[
\lim_{k\to \infty}\tilde{x}_i(k)=(x_i(k)-x_r(k))\to 0
\] 
i.e. $x_i(k)\to x_r(k)$ as $k\to \infty$, which proves the result.
\end{proof}
\begin{figure}[h]
	\includegraphics[width=8.5cm, height=5cm]{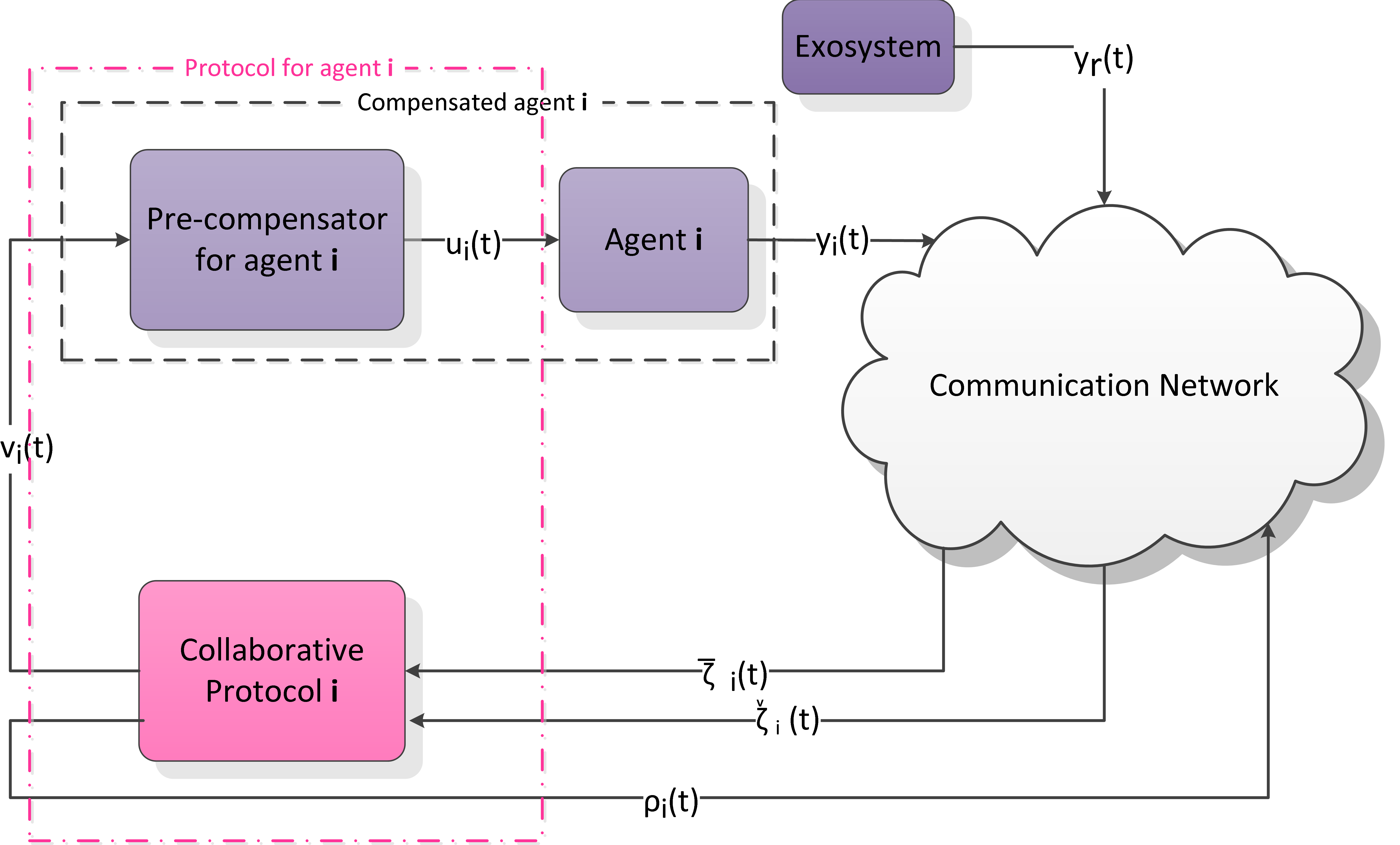}
	\centering
	\caption{Architecture of Protocol $4$}\label{heterogeneous_regulated}
\end{figure}

\section{Numerical Example}
In this section, we will provide a numerical example for state synchronization of homogeneous MAS with partial-state coupling.

Consider agent models for $i=\{1, \hdots, N\}$ as
\begin{equation}
\begin{system*}{cl}
{x}_i(k+1)&=\begin{pmatrix}
0.5&1&1\\0&0.866&-0.5\\0&0.5&0.866
\end{pmatrix}x_i(k)+\begin{pmatrix}
0\\0\\1
\end{pmatrix}u_i(k),\\
y_i(k)&=\begin{pmatrix}
1&0&0
\end{pmatrix}x_i(k)
\end{system*}
\end{equation}
We choose $H\T=\begin{pmatrix}
1.4327&0.4143&0.6993
\end{pmatrix}$ and $K=\begin{pmatrix}
0.0695&1.7625&1.2051
\end{pmatrix}$. Therefore our scale-free collaborative \emph{Protocol $2$} would be as follows.

 \begin{equation}\label{Ex-pscp2}
	\begin{system}{cll}
		{\eta}_i(k+1) &= \begin{pmatrix}
			    0.5  &  1  &  1\\
			0   & 0.866  & -0.5\\
			-0.0695  & -1.2625 &  -0.3391
		\end{pmatrix}\eta_i(k)\\
		&+\begin{pmatrix}
		0.5&1&1\\0&0.866&-0.5\\0&0.5&0.866
		\end{pmatrix}({\hat{x}}_i(k)-\hat{\zeta}_i(k)) \\
		{\hat{x}}_i(k+1) &= \begin{pmatrix}
		   -0.9327  &  1  &  1\\
		-0.4143  &  0.866  & -0.5\\
		-0.6993  &  0.5  &  0.866
		\end{pmatrix}\hat{x}_i(k)\\
		&-\begin{pmatrix}
			         0    &     0    &     0\\
			0      &   0    &     0\\
			0.0695  &  1.7625  &  1.2051
		\end{pmatrix}\hat{\zeta}_i(k)+\begin{pmatrix}
		1.4327\\0.4143\\0.6993
		\end{pmatrix}{\zeta}_i(k) \\
		u_i(k) &= - \begin{pmatrix}
		0.0695&1.7625&1.2051
		\end{pmatrix}\eta_i(k)
	\end{system}
\end{equation}

Now we are creating three homogeneous MAS with different number of agents and different communication topologies to show that the designed collaborative protocol \eqref{Ex-pscp2}  is independent of the communication network and number of agents $N$.
		\begin{figure}[h]
	\includegraphics[width=8.5cm, height=3cm]{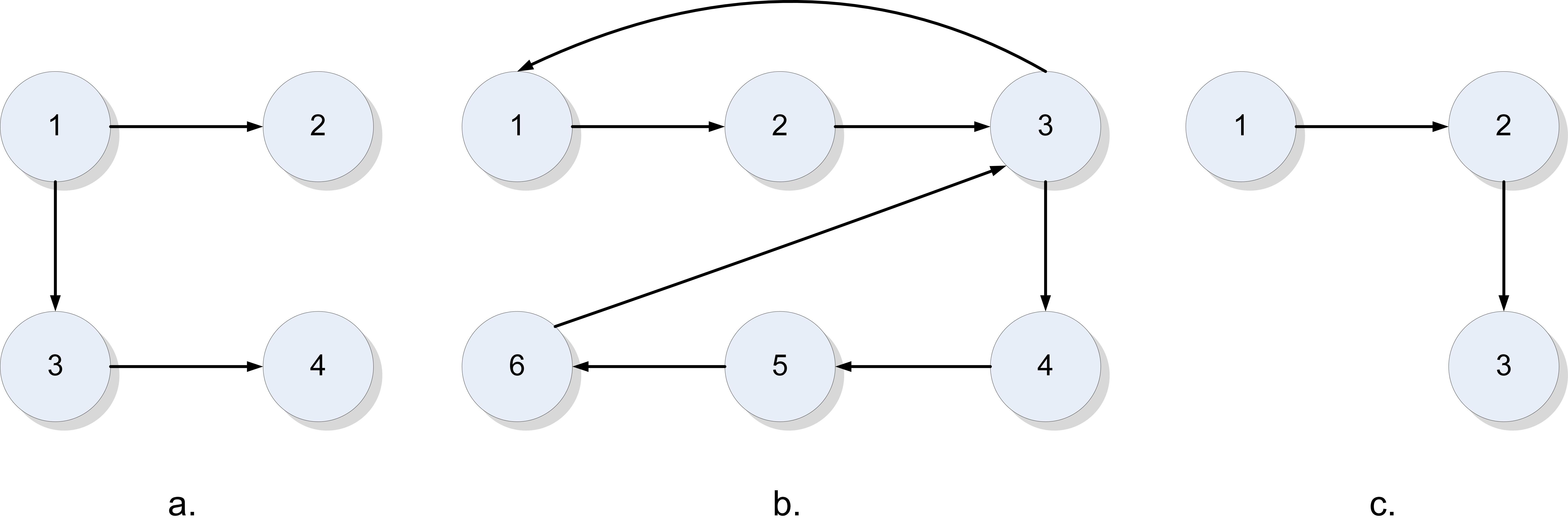}
	\centering\vspace*{-2mm}
	\caption{Communication graph for a. case I, b. case II, c. case III}\label{Drawing3}
\end{figure}
\begin{itemize}
	\item \textit{Case I}: Consider MAS with 4 agents $N=4$, and directed communication topology shown in Figure \ref{Drawing3}.a.
	
%
	\item \textit{Case II}: In this case, we consider MAS with 6 agents $N=6$, and directed communication topology shown in Figure \ref{Drawing3}.b.

	\item \textit{Case III}: Finally, we consider the MAS with $3$ agents, $N=3$ and communication graph shown in Figure \ref{Drawing3}.c.
		
%
	
\end{itemize}

The results are demonstrated in Figure \ref{results_case11}-\ref{results_case33}. The simulation results show that the protocol design \eqref{Ex-pscp2} is independent of the communication graph and is scale-free so that we can achieve synchronization with one-shot protocol design, for any graph with any number of agents.

\begin{figure}[t] 
	\includegraphics[width=8cm, height=5cm]{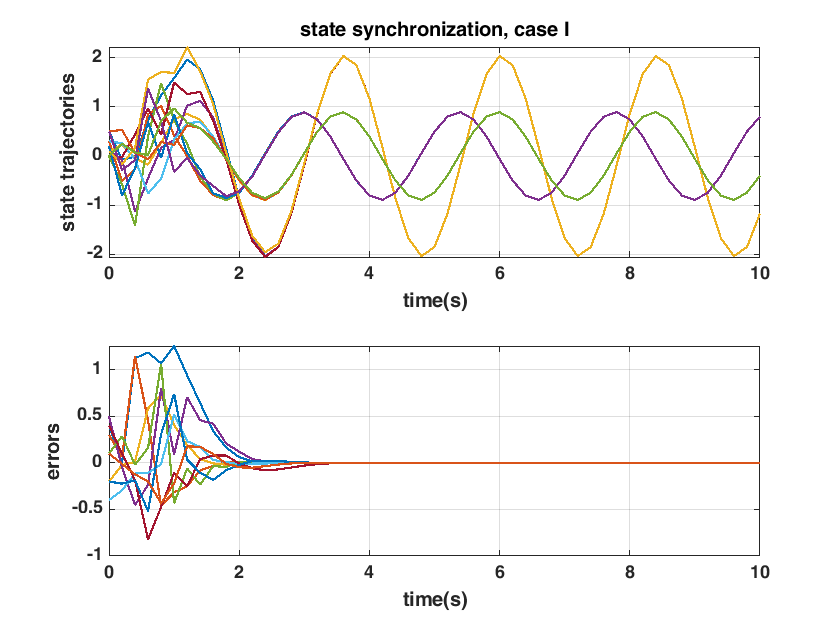}
	\centering
	\caption{State synchronization for MAS with communication graph $I$.}\label{results_case11}\vspace*{-2mm}
\end{figure}
\begin{figure}[t]
	\includegraphics[width=8cm, height=5cm]{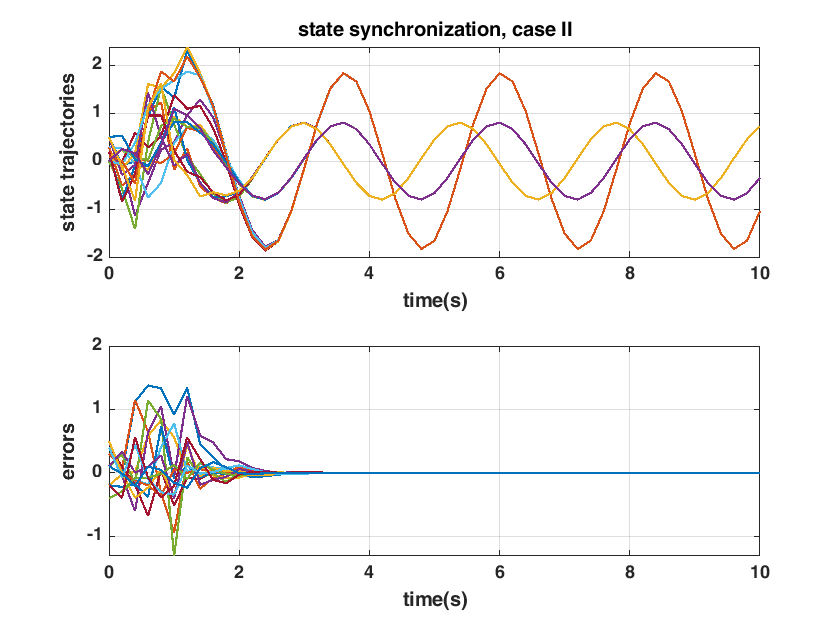}
	\centering
	\caption{State synchronization for MAS with communication graph $II$.}\label{results_case22}\vspace*{-2mm}
\end{figure}
\begin{figure}[t!]
	\includegraphics[width=8cm, height=5cm]{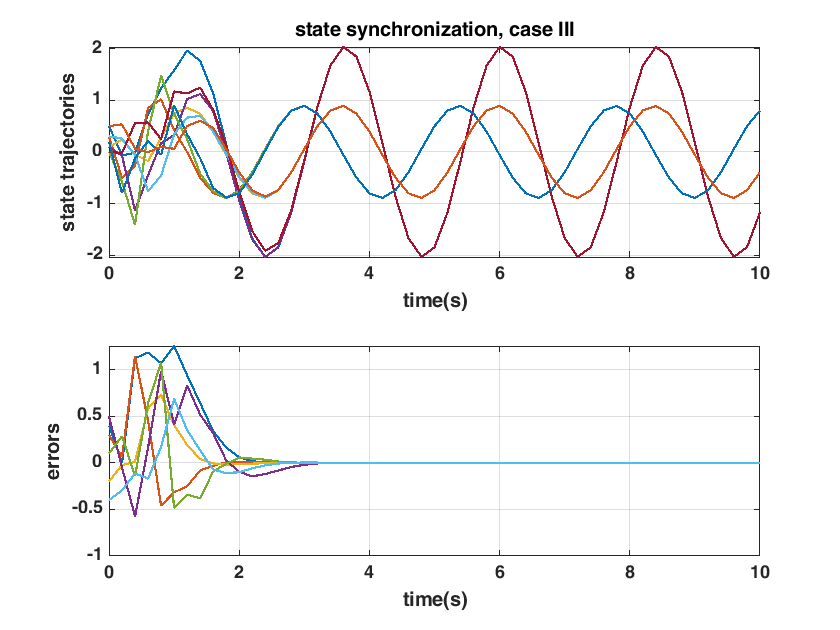}
	\centering
	\caption{State synchronization for MAS with communication graph $III$.}\label{results_case33}\vspace*{-2mm}
\end{figure}
	
\bibliographystyle{plain}
\bibliography{referenc}
\end{document}